\documentclass[aps, 11pt, superscriptaddress, preprint, longbibliography]{revtex4}

\usepackage[dvipdfmx]{graphicx}

\usepackage[utf8]{inputenc}
\usepackage[T1]{fontenc}
\usepackage{mathptmx}
\usepackage{etoolbox}

\usepackage{amsmath}
\usepackage{amssymb}
\usepackage{amsthm}
\usepackage{braket}
\usepackage{bm}
\usepackage{ulem}

\usepackage{color}

\newtheorem{lemma}{Lemma}
\newtheorem{remark}{Remark}
\newtheorem{theorem}{Theorem}

\makeatletter
\def\@email#1#2{%
 \endgroup
 \patchcmd{\titleblock@produce}
  {\frontmatter@RRAPformat}
  {\frontmatter@RRAPformat{\produce@RRAP{*#1\href{mailto:#2}{#2}}}\frontmatter@RRAPformat}
  {}{}
}%
\makeatother

\begin{document}

\title{Koopman spectral analysis of elementary cellular automata}

	\author{Keisuke Taga}
	\email{tagaksk@akane.waseda.jp}
	\affiliation{Department of Physics, School of Advanced Science and Engineering, Waseda University, Tokyo 169-8555, Japan}

	\author{Yuzuru Kato}
	\affiliation{Department of Systems and Control Engineering,
		Tokyo Institute of Technology, Tokyo 152-8552, Japan}

	\author{Yoshinobu Kawahara}
	\affiliation{Institute of Mathematics for Industry, Kyushu University, 
	and Center for Advanced Intelligence Project, RIKEN, Fukuoka 819-0395, Japan}

	\author{Yoshihiro Yamazaki}
	\affiliation{Department of Physics, School of Advanced Science and Engineering, Waseda University, Tokyo 169-8555, Japan}

	\author{Hiroya Nakao}
	\affiliation{Department of Systems and Control Engineering,
	Tokyo Institute of Technology, Tokyo 152-8552, Japan}
	\date{\today}

\begin{abstract}
We perform a Koopman spectral analysis of elementary cellular automata (ECA).
By lifting the system dynamics using a one-hot representation of the system state, we derive a matrix representation of the Koopman operator as the transpose of the adjacency matrix of the state-transition network.
The Koopman eigenvalues are either zero or on the unit circle in the complex plane, and the associated Koopman eigenfunctions can be explicitly constructed.
From the Koopman eigenvalues, we can judge the reversibility, determine the number of connected components in the state-transition network, evaluate the period of asymptotic orbits, and derive the conserved quantities for each system.
We numerically calculate the Koopman eigenvalues of all rules of ECA on a one-dimensional lattice of 13 cells with periodic boundary conditions. It is shown that the spectral properties of the Koopman operator reflect Wolfram’s classification of ECA.

\end{abstract}

\maketitle
\textbf{The Koopman operator theory for analyzing nonlinear dynamical systems has attracted much attention recently. In this paper, we introduce a Koopman spectral analysis of elementary cellular automata (ECA), which have been studied as models of real-world systems such as pigmentation patterns on shells, peeling patterns of adhesive tapes, and congestion dynamics of traffic flows. We derive a finite-dimensional representation of the Koopman operator for ECA and construct the Koopman eigenvalues and Koopman eigenfunctions explicitly. The Koopman analysis of ECA can reveal fundamental properties of the system, such as reversibility and conserved quantities. We also numerically calculate the Koopman eigenvalues of all rules of ECA of 13 cells with periodic boundary conditions and show that the spectral properties reflect Wolfram’s classification of ECA. Our Koopman analysis of ECA can provide insights into the relationship between the spatiotemporal dynamics and spectral properties of ECA and also serve as a testbed for examining various algorithms of dynamic mode decomposition (DMD) for estimating the Koopman eigenvalues using timeseries data observed from spatially extended systems.}

\section{INTRODUCTION}

The Koopman operator theory~\cite{koopman1931hamiltonian,vonneumann1932,mezic2005spectral,budivsic2012applied,mauroy2020koopman,kutz2016dynamic,bollt2018matching} 
for analyzing nonlinear dynamical systems has attracted much attention recently. For a given dynamical system, the Koopman operator analysis focuses on the evolution of observables rather than the system state itself.
As the Koopman operator is linear even if the system dynamics is nonlinear, it provides a method to analyze nonlinear dynamical systems by using standard spectral methods for linear systems.
On the other hand, because the Koopman operator acts on a function space of observables, it can be generally infinite-dimensional and difficult to analyze.
Analytical solutions to the eigenvalue problems of the Koopman operator can be obtained only in a limited class of linearizable or exactly solvable systems. Koopman spectral analysis of partial differential equations (PDEs) describing spatiotemporal patterns has also been considered, but analytical results can be obtained only for solvable PDEs~\cite{nathan2018applied,page2018koopman,nakao2020spectral,parker2020koopman}.
In this study, as the simplest system exhibiting spatiotemporal dynamics, we introduce the Koopman operator analysis to elementary cellular automata (ECA)~\cite{moore1962machine,wolfram1983statistical,wolfram2002new,martinez2013note,coombes2009geometry,nishinari1998analytical,kari2005theory}. 
ECA are described by discrete space-time variables and binary state variables, which can exhibit rich dynamics including periodic and chaotic ones despite the simplicity of the evolution rules.
The pattern dynamics of all 256 rules of ECA have been classified qualitatively into four classes by Wolfram~\cite{wolfram2002new}.
Some of the ECA rules have also been studied as models of real-world systems such as pigmentation patterns on shells, peeling patterns of adhesive tapes, and congestion dynamics of traffic flows~\cite{wolfram2002new,coombes2009geometry,nishinari1998analytical,ohmori2019comments}.

Because the state space of ECA on a finite lattice is finite-dimensional, we can explicitly represent the Koopman operator of ECA by a finite-size matrix.
This allows us to rigorously derive fundamental relationships between the dynamical properties of the ECA and the spectral properties of the Koopman operator, and also to numerically obtain the complete set of the Koopman eigenvalues of ECA for small systems.
These exact results will also serve as a testbed for examining various algorithms of dynamic mode decomposition (DMD) for estimating the Koopman eigenvalues using time-series data observed from spatially extended dynamical systems~\cite{schmid2010dynamic,rowley2009spectral,tu2013dynamic,kutz2016dynamic,kawahara2016dynamic,korda2017data,takeishi2017learning,mauroy2020koopman}.

This paper is organized as follows. In Sec. II, we briefly explain ECA and introduce the one-hot representation. In Sec. III, we introduce the Koopman operator for ECA and derive the matrix representation. In Sec. IV, we explicitly construct the Koopman eigenfunctions. In Sec. V, we perform a thorough Koopman spectral analysis for all rules of ECA with 13 cells. Section VI discusses the number of eigenvalues and dynamic mode decomposition, and section VII gives a summary. In Appendix, we provide a brief discussion on the Perron-Frobenius operator.

\section{ElLEMENTARY CELLULAR AUTOMATA}

\subsection{Wolfram's classification}

A cellular automaton is a dynamical system whose space, time, and state variables are all discrete. In this study, we analyze ECA, the simplest case of the cellular automata on a one-dimensional lattice~\cite{wolfram2002new}.
In ECA, the state of each cell on the lattice takes either $0$ or $1$ and evolves with time under a given rule. The next state of each cell is determined by the current states of itself and the two neighboring cells.
Despite the simplicity of the rules, ECA can exhibit various pattern dynamics including homogeneous, pulse-like, traveling, oscillating, chaotic, and complex behaviors~\cite{wolfram2002new}.
Some of the ECA rules are directly associated with partial differential equations (PDEs) describing real-world systems. For example, rule 18 of ECA exhibits Sierpinski-gasket patterns similar to those of the Gray-Scott PDE for a chemical reaction~\cite{coombes2009geometry}, and rule 184 of ECA, which can be derived from the Burgers PDE by ultradiscretization, is discussed as a model of traffic flows~\cite{tokihiro1996soliton,nishinari1998analytical}.

\begin{figure}[h!]
    \centering
    \includegraphics[width=0.85\hsize]{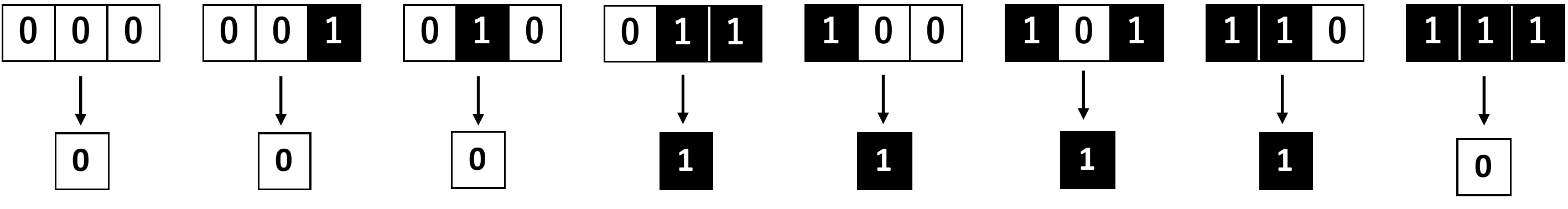}
	\caption{Rule 120 of ECA.}
    \label{fig1}
\end{figure}

In Fig.~\ref{fig1}, rule 120 of ECA on a one-dimensional lattice
is shown as an example, where we assign $0$ to white cells and $1$ to black cells, respectively.
The states of the cells at the next time step can be specified by the $8$ digits, which gives $120$ in decimal numbers when interpreted as an $8$-bit binary number.
Thus, ECA on a one-dimensional lattice possess $2^{2^3} = 2^8=256$ possible rules in total.
Among them, there are equivalent rules that can be transformed to each other by simple transformations, e.g., by spatial reflection; hence, $88$ rules are mutually independent~\cite{Li1990equivalent}.

Wolfram qualitatively classified these rules into four classes based on their typical dynamics in sufficiently large systems as follows~\cite{wolfram2002new}:
\begin{itemize}{}
\item Class I: The system converges to a homogeneous stationary state.

\item Class II: The system converges to an inhomogeneous stationary state or a periodic state.

\item Class III: The system exhibits chaotic dynamics.

\item Class IV: The system exhibits a complex mixture of regular and chaotic dynamics.

\end{itemize}
Here, the term chaotic or complex describes the visual appearance of the system dynamics.
The chaotic properties of the system can also be characterized in the sense that two nearby system states separate from each other exponentially~\cite{grassberger1986long,bagnoli1992damage}.
As the state space of ECA on a finite lattice is finite, all orbits eventually behave periodically also in class III and class IV, but their periods can be extremely large.

Wolfram’s classification captures the main characteristics of ECA, but it is essentially qualitative as it is based on visual inspection of the spatiotemporal patterns.
Thus, various criteria for more quantitative classification of ECA have been proposed~\cite{kari2005theory,martinez2013note}.
In this study, we analyze the spectral property of the Koopman operator of ECA and discuss their relationship with Wolfram’s
classification.

\subsection{State space of ECA}

The ECA can be considered a discrete dynamical system of the form
\begin{align}
    {\bm x}_{n+1} = {\bm F}({\bm x}_n),
\end{align}
where ${\bm x}_n \in M$ represents the system state at discrete time $n$ and the map ${\bm F} : M \to M$ represents the dynamics of the system state.
When the system has $N$ cells, the state space $M$ is a set of $N$-dimensional binary vectors, i.e., $M = \{ \bm{x} = (x_1, ..., x_N) \ | \ x_{1, ..., N} \in \{0, 1\} \}$.
The total number of possible system states is $|M| = 2^N$.

We can also assign each system state ${\bm x} = (x_1, ..., x_N) \in M$ an index $q \in \{1, ..., 2^N\}$ as 
\begin{align}
q = \sum_{j=1}^N 2^{j-1} x_j + 1.
\end{align}
For example, the state ${\bm x} = (0, 1, 0)$ of 3-cell ECA is indexed as $q=3$.
Using this index, we can also represent the system state as an element in the set of $2^N$ integers, $\{ 1, ..., 2^N \}$.
We denote the system state in $M$ with index $q$ as ${\bm x}^{(q)}$ in the following discussion.

The transitions between the system states can be represented by a directed {\it state-transition network}, as illustrated in Fig.~\ref{fig2} for rule $120$ on the minimal lattice with $3$ cells and periodic boundary conditions (the leftmost and rightmost cells are adjacent to each other), where the network nodes represent individual system states and directed links represent the transitions between the states determined by the ECA rule.    
The states such as ${\bm x} = (1,1,0)$, $(1,0,1)$, and $(0,1,1)$, which cannot be reached from any states, are called the {\it garden-of-Eden states}~\cite{moore1962machine,kari2005theory} (sources  without self loops).
ECA without garden-of-Eden states are known to be {\it reversible}~\cite{kari2005theory}, 
namely, the map ${\bm F}$ is bijective and there exists another rule represented by ${\bm F}^{-1}$ such that ${\bm x} = {\bm F}^{-1}({\bm F}({\bm x})) = {\bm F}({\bm F}^{-1}({\bm x}))$ for all ${\bm x} \in M$.

As illustrated in Fig.~\ref{fig2}, each network of ECA consists of several weakly connected components (maximal connected subnetworks), i.e., the maximal sets of states that are connected with each other when the directions of the links are neglected. 
In each connected component, only a single periodic orbit (including a period-$1$ stationary state) exists, and the system states not included in this periodic orbit belong to a tree subnetwork whose root is a state in the periodic orbit and whose leaves (terminal nodes) are garden-of-Eden states.
For example, there are two connected components A and B in Fig.~\ref{fig2}, each with a single periodic orbit of period 3 (A) or 1 (B), and the component B on the right has a period-$1$ orbit and a tree subnetwork with three garden-of-Eden states.
This is because the rule of ECA is deterministic and each state has only a single outgoing link; hence, no more than two periodic orbits can exist in the same connected component.

The system states not included in the periodic orbit converge to the periodic orbit in the same connected component within finite steps, because if an orbit that does not converge to the periodic orbit exists, the corresponding time series should include overlapping states when it becomes longer than $2^N$ by the pigeonhole principle.
Thus, the periodic orbit is an attractor and the connected component of the network is the basin of attraction.

\begin{figure}[b]
\begin{center}
\includegraphics[width=0.85\hsize]{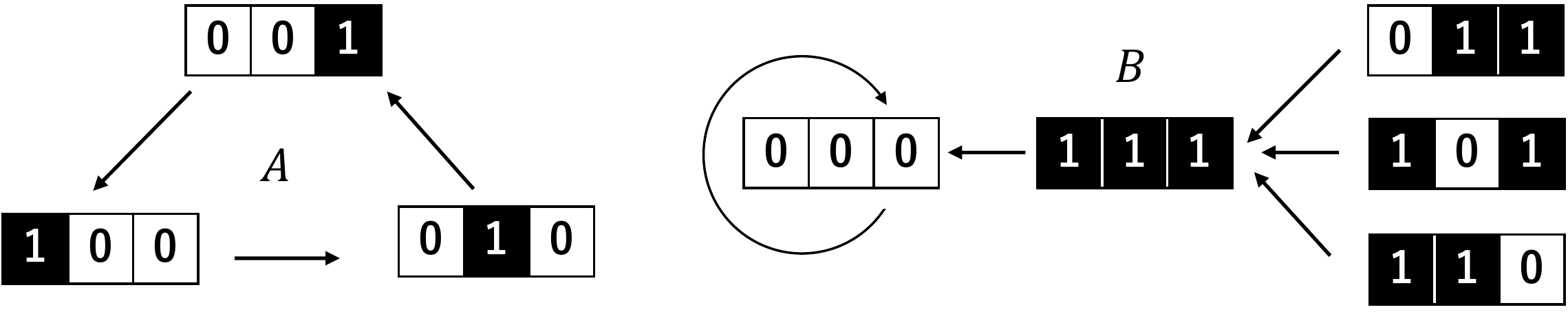}
\caption{State-transition network of ECA with 3 cells and periodic boundary conditions for rule 120, which comprises two subnetworks (connected components).}
\label{fig2}
\end{center}
\end{figure}

\section{KOOPMAN SPECTRAL ANALYSIS OF ECA}

\subsection{Koopman operator}

We now introduce the Koopman operator for ECA.
An observable $G : M \to {\mathbb C}$ is a function that maps the system state ${\bm x} \in M$ to an observed value $G({\bm x}) \in {\mathbb C}$. 
The Koopman operator $\hat{K}$ describes the evolution of the observable $G$ and is defined as 
\begin{align}
(\hat{K} G)({\bm x}) = G \circ {\bm F}({\bm x})=G({\bm F}({\bm x}))
\end{align}
for ${\bm x} \in M$, where $\circ$ denotes composition of functions.
This equation means that instead of evolving the system state as ${\bm x} \to {\bm F}({\bm x})$ and observe it by a fixed observable $G$ to obtain the measurement results $G( {\bm x} ) \to G( {\bm F}({\bm x}) )$, we consider that the system state is fixed at some ${\bm x}$ but rather the observable $G$ evolves such that the same measurement results $G( {\bm x} ) \to (\hat{K} G)( {\bm x} )$ are obtained.

For systems with continuous states, the space of the observables can be infinite-dimensional. However, for ECA with $N$ cells, the total number of the system states is $2^N$; hence, the Koopman operator can be represented by a finite-dimensional $2^N \times 2^N$ matrix, as we show in Subsection III B.

The linearity of the Koopman operator follows from the definition; specifically, 
\begin{align}
\{ \hat{K}( a G + b H ) \} ({\bm x})
&= a G ({\bm F} ({\bm x})) + b H ({\bm F}({\bm x})) 
\cr
&= a(\hat{K} G)({\bm x}) + b (\hat{K} H)({\bm x})
\end{align}
holds for arbitrary observables $G$ and $H$ and complex numbers $a$ and $b$.
Thus, even if the dynamics ${\bm F}$ is nonlinear, the corresponding evolution of the observable is linear and the spectral methods for linear systems can be used to analyze the system's dynamical properties.

In particular, by solving the eigenvalue equation for $\hat{K}$ (assuming $\hat{K}$ has only a discrete spectrum, which is the case for ECA),
\begin{align}
(\hat{K} \Phi) ({\bm x}) = \lambda \Phi({\bm x}),
\label{koopmaneigen}
\end{align}
we can obtain a set of eigenvalues and associated eigenfunctions, $\{ \lambda_\alpha,\ \Phi^{(\alpha)} \}$, where $\alpha = 1, ..., 2^N$ is the index.
As we explain later, the system may not possess $2^N$ independent eigenfunctions. If so, we also consider the generalized eigenfunctions with higher ranks, which satisfy $(\hat{K} - \lambda \hat{I})^{m+1} {\Phi}({\bm x}) = 0$ and $(\hat{K} - \lambda \hat{I})^{m} {\Phi}({\bm x}) \neq 0$, where $m+1 \geq 2$ is the rank and $\hat{I}$ is an identity operator.

The Koopman eigenfunctions play essentially important roles in the Koopman spectral analysis of nonlinear dynamical systems and can represent physical information of the system such as conserved quantities~\cite{mezic2005spectral,budivsic2012applied,kutz2016dynamic,bollt2018matching,mauroy2020koopman,nakao2020spectral}.
For example, the Koopman eigenfunction $\Phi({\bm x})$ with the eigenvalue $\lambda = 1$ yields a {\it conserved quantity} (invariant) of the system because
\begin{align}
(\hat{K} \Phi)({\bm x}) = \Phi({\bm F}({\bm x})) = \Phi({\bm x}).
\end{align}
Thus, we have $\Phi({\bm x}_n) = (\hat{K}^n \Phi)({\bm x}_0) = \Phi({\bm x}_0) = const.$ for all $n \geq 0$.
We note that we consider general conserved quantities here, not the additive conserved quantities discussed in the context of the statistical mechanical approach to ECA~\cite{takesue1987reversible,hattori1991additive}.

\subsection{Matrix representation of the Koopman operator}

Following Budi\v{s}i\'{c}, Mohr, and Mezi\'{c}'s argument for cyclic groups~\cite{budivsic2012applied}, we introduce $2^N$ indicator functions, $b_q: M \to \{0, 1\}$ for $q=1, ..., 2^N$, where
\begin{align}
b_q( {\bm x} ) =
\left\{\begin{aligned}
&1 \quad \mbox{if ${\bm x} = {\bm x}^{(q)}$},
\cr
&0 \quad \mbox{otherwise}.
\end{aligned}\right.
\end{align}
This gives a {\it 'one-hot' representation} of the system state ${\bm x}$; specifically, if  ${\bm x}$ takes the $q$th state ${\bm x}^{(q)}$,
only a single indicator function $b_q({\bm x})$ among $\{ b_1({\bm x}), ..., b_{2^N}({\bm x}) \}$ takes the value $1$ and all the other indicator functions $b_r({\bm x})$ with $r \neq q$ are $0$.

The evolution of these indicator functions, which are also observables of the system, is expressed as
\begin{align}
(\hat{K} b_q) ( {\bm x} ) = b_q( {\bm F}( {\bm x} ) ) = \sum_{r=1}^{2^N} A_{qr} b_r(  {\bm x} )
\end{align}
for $q=1, ..., 2^N$. 
Here, $\hat{K}$ is the Koopman operator and an adjacency matrix $A \in \{0,1\}^{2^N \times 2^N} $ of the system states is introduced, whose components are given by
\begin{eqnarray}
A_{qr}=\left\{\begin{aligned}
1 &\quad \mbox{if\ the system evolves from state $r$ to $q$},\\
0 &\quad \mbox{otherwise}
\end{aligned}\right.
\label{adjacency}
\end{eqnarray}
for $q, r = 1, ..., 2^N$.
This adjacency matrix $A$ represents the state-transition network of the system as illustrated in Fig.~\ref{fig2}.
For example, when the system state ${\bm x}_n$ is in the $p$th state  at time $n$, i.e., ${\bm x}_n = {\bm x}^{(p)}$, we have $b_p({\bm x}_n) = 1$ and $b_r({\bm x}_n) = 0$ for $r \neq p$. If this state evolves into the $q$th state at time $n+1$, i.e., if ${\bm x}_{n+1} = {\bm F}({\bm x}_n={\bm x}^{(p)}) = {\bm x}^{(q)}$ and $A_{q p} = 1$, we have $b_q({\bm x}_{n+1}) = 1$ and $b_r({\bm x}_{n+1}) = 0$ for $r \neq q$.
It is noted that only a single component of each column of $A$ takes $1$ and all other components are $0$ because ECA are deterministic.

For a reversible system, we can obtain the inverse dynamics by reversing the directions of all links in the state-transition network, which can be done because the map ${\bm F}$ is bijective. Therefore, the inverse of $A$ is given by its transpose, i.e., $A^{-1} = A^{\sf T}$, where $\sf{T}$ represents matrix transposition, ;namely,
\begin{align}
(A A^{\sf T})_{qs} = \sum_{r=1}^{2^N} A_{qr} A^{\sf T}_{rs} = \sum_{r=1}^{2^N} A_{qr} A^{-1}_{rs} = \delta_{qs}.
\label{Aunitary}
\end{align}
Since $A$ is real, $A^{\sf T}_{rs} = A^*_{rs}$ holds, where $*$ indicates conjugate transpose. Therefore, $A$ is a unitary matrix.

The indicator functions form a basis of the space of observables~\cite{budivsic2012applied}.
A general observable $G({\bm x})$ can be expressed by using the indicator functions as
\begin{align}
G({\bm x}) = \sum_{q=1}^{2^N} g_q b_q({\bm x}),
\label{expansionobservable}
\end{align}
where the coefficient $g_q \in {\mathbb C}$ represents the measurement outcome when ${\bm x} = {\bm x}^{(q)}$ ($q=1, ..., 2^N$).
The evolution of this observable $G$ is given by
\begin{align}
(\hat{K} G)({\bm x}) 
&= \sum_{q=1}^{2^N} g_q (\hat{K} b_q)({\bm x})
= \sum_{q=1}^{2^N} g_q \left( \sum_{r=1}^{2^N} A_{qr} b_r({\bm x}) \right)
\cr
&=
\sum_{q=1}^{2^N} \left( \sum_{r=1}^{2^N} A^{\sf T}_{qr} g_r \right) b_q({\bm x}).
\label{koopmanmatrix0}
\end{align}
Therefore, the matrix representation of the evolution of the observable $G$ is given by
\begin{align}
{K} \begin{pmatrix} g_1 \\ \vdots \\ g_{2^N} \end{pmatrix} 
= 
\begin{pmatrix}
A_{1, 1} &  \cdots & A_{2^N, 1} \\
\vdots &             & \vdots \\
A_{1, 2^N} & \cdots & A_{2^N, 2^N}
\end{pmatrix}
\begin{pmatrix} g_1 \\ \vdots \\ g_{2^N} \end{pmatrix},
\label{koopmanmatrix}
\end{align}
where ${K}$ represents the evolution of the vector of coefficients ${\bm g} = (g_1, ..., g_{2^N})^{\sf T}$ of the observable $G$.
Thus, the transposed adjacency matrix $A^{\sf T}$ gives a matrix representation $K$ of the Koopman operator $\hat{K}$, which we call the {\it Koopman matrix} in what follows.

By expressing the Koopman eigenfunction $\Phi({\bm x})$ associated with the eigenvalue $\lambda$ as
\begin{align}
\Phi({\bm x}) =  \sum_{q=1}^{2^N} \phi_q b_q({\bm x})
\label{koopmaneigenfunction}
\end{align}
and plugging into the eigenvalue equation~(\ref{koopmaneigen}), we obtain the eigenvalue equation for the matrix ${K}$,
\begin{align}
{K} \begin{pmatrix} \phi_1 \\ \vdots \\ \phi_{2^N} \end{pmatrix} = \lambda  \begin{pmatrix} \phi_1 \\ \vdots \\ \phi_{2^N} \end{pmatrix}.
\label{koopmaneigenequation}
\end{align}
By solving the above equation, we can obtain the set of Koopman eigenvalues and {\it Koopman eigenvectors} $\{ \lambda_\alpha,\ {\bm \phi}^{(\alpha)} = (\phi_1^{(\alpha)}, ..., \phi_{2^N}^{(\alpha)})^{\sf T} \}$ and construct the Koopman eigenfunction $\Phi^{(\alpha)}$ from the obtained Koopman eigenvector ${\bm \phi}^{(\alpha)}$.
When the system is reversible, the Koopman matrix $K = A^{\sf T}$ is unitary because $A$ is unitary as shown in Eq.~(\ref{Aunitary}); thus, all eigenvalues of $K$ are on the unit circle in the complex plane.

We note that the matrix $K$ may not be diagonalizable and $2^N$ independent eigenvectors satisfying Eq.~(\ref{koopmaneigenequation}) may not be obtained.
In such cases, to obtain a basis set of $2^N$ independent vectors, we also need to consider the generalized eigenvectors with rank $m+1 \geq 2$, satisfying $(K - \lambda_\alpha I)^{m+1} {\bm \phi}^{(\alpha, m+1)} = 0$ and $(K - \lambda_\alpha I)^{m} {\bm \phi}^{(\alpha, m)} \neq 0$
when the eigenvalue $\lambda_\alpha$ is defective, where $I$ is an identity matrix.
We can then construct the generalized Koopman eigenfunction of $\hat{K}$ in the form $\Phi^{(\alpha, m+1)}({\bm x}) = \sum_{q=1}^{2^N} \phi^{(\alpha, m+1)}_q b_q({\bm x})$, satisfying $( \hat{K} - \lambda \hat{I} )^{m+1} \Phi^{(\alpha, m+1)}({\bm x}) = 0$ as mentioned previously. 

Thus, we have derived a matrix representation of the Koopman operator and observables.
We note here that the state space of the system is lifted from the $N$-dimensional space of binary vectors to a $2^N$-dimensional space of integer indices by the introduction of the one-hot representation.
This gives a linear representation of the system dynamics characterized by the adjacency matrix $A$ in Eq.~(\ref{adjacency}) and the linear evolution of the observable characterized by the Koopman matrix $K$ in Eq.~(\ref{koopmanmatrix}).

\begin{figure}[h!]
    \centering
    \includegraphics[width=1\hsize]{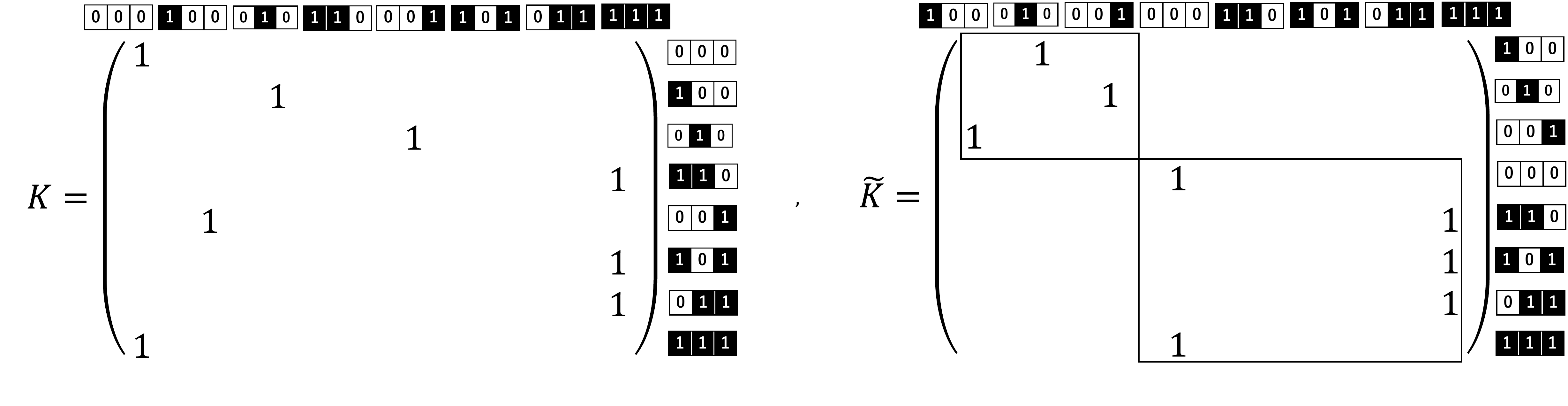}
	\caption{Koopman matrix $K$ and block-diagonalized $\tilde{K}$ for rule 120.}
    \label{fig3}
\end{figure}

\subsection{Example: rule 120 with 3 cells}

As a simple example, we analyze rule 120 of ECA with 3 cells and periodic boundary conditions shown in Fig.~\ref{fig2}.
The matrix $K$ in Fig.~\ref{fig3} shows the Koopman matrix ${K} = A^{\sf T}$ of rule 120 obtained from the state-transition network in Fig.~\ref{fig2}, where the matrix components with $0$ are left blank and only those with $1$ are shown.
We can block-diagonalize $K$ as
\begin{align}
\tilde{K} = \begin{pmatrix} K_A & 0 \\ 0 & K_B \end{pmatrix}
\end{align}
by reordering the rows and columns as shown in Fig.~\ref{fig3}, where $K_A \in \{0, 1\}^{3 \times 3}$ and $K_B \in \{0, 1\}^{5 \times 5}$. This is because the state-transition network can be divided into two- independent connected components, i.e., subnetwork $A$ consisting of states $\{2, 3, 5\}$ and subnetwork $B$ consisting of $\{1, 4, 6, 7, 8\}$, as shown in Fig.~\ref{fig2}. As the time evolutions of the states belonging to different subnetworks are independent, we can analyze each subnetwork separately.

\subsubsection{Subnetwork $A$}

Let us consider the subnetwork $A$ in Fig~\ref{fig2}.
As shown in Fig~\ref{fig4}(a), the system has three Koopman eigenvalues; i.e., $\lambda_{1} = 1$, $\lambda_{2} = e^{2 \pi i / 3}$, and $\lambda_{3}=e^{4 \pi i / 3}$.
These three eigenvalues on the unit circle in the complex plane correspond to the period-3 orbit, ${\bm x} =  (1, 0, 0) (q=2) \to (0, 1, 0) (3) \to (0, 0, 1) (5)$ in Fig.~\ref{fig2}.
\begin{figure}[h!]
    \centering
    \includegraphics[width=\hsize]{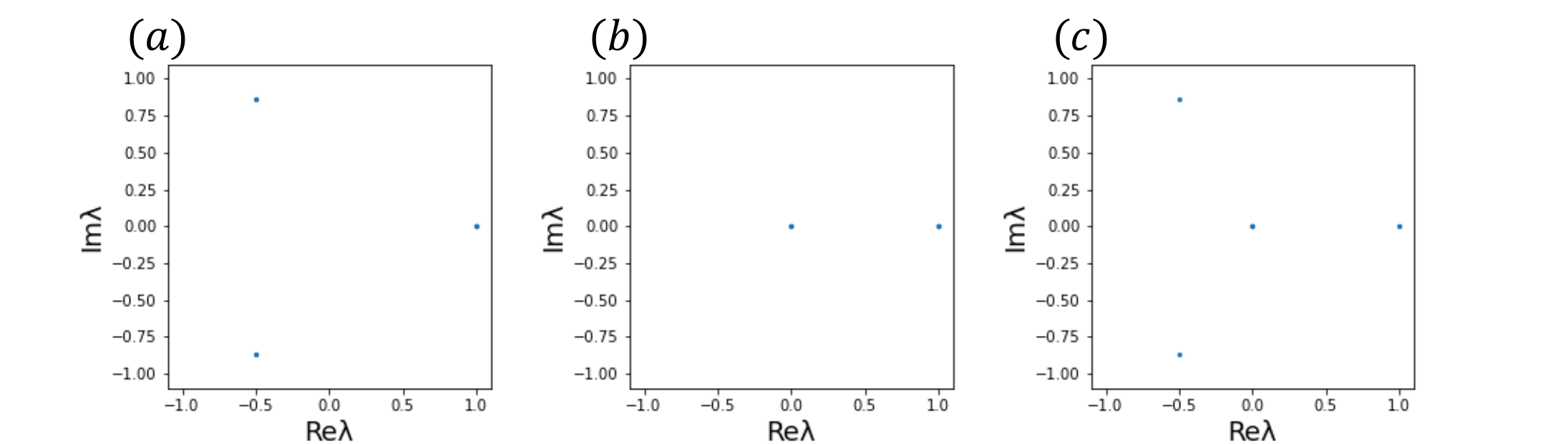}
    \caption{Koopman eigenvalues of rule 120 with 3 cells and periodic boundary conditions. (a) Eigenvalues from subnetwork A. (b) Eigenvalues from subnetwork B. (c) All eigenvalues of the whole network.}
    \label{fig4}
\end{figure}

In this case, the Koopman submatrix $K_A$ corresponding to subnetwork $A$ is diagonalizable. We find a Koopman eigenvector of $K_A$ associated with $\lambda_{1} = 1$, 
\begin{align}
{\bm \phi_A}^{(1)}=(\phi^{(1)}_2, \phi^{(1)}_3, \phi^{(1)}_5)^{\sf T}
=
(1, 1, 1)^{\sf T},
\end{align}
and also two other eigenvectors associated with $\lambda_{2} = e^{2 \pi i / 3}$ and $\lambda_{3} = e^{4 \pi i / 3}$,
\begin{align}
{\bm \phi_A}^{(2)}=(\phi^{(2)}_2, \phi^{(2)}_3, \phi^{(2)}_5)^{\sf T}
&=
(e^{\frac{2 i \pi }{3}}, e^{\frac{4 i \pi }{3}}, 1)^{\sf T},
\;\;\;
{\bm \phi}_A^{(3)}=(\phi^{(3)}_2, \phi^{(3)}_3, \phi^{(3)}_5)^{\sf T}
=
(e^{\frac{4 i \pi }{3}}, e^{\frac{2 i \pi }{3}}, 1)^{\sf T}.
\end{align}
The set of these three vectors $\{ {\bm \phi}_A^{(1)}, {\bm \phi}_A^{(2)}, {\bm \phi}_A^{(3)} \}$ corresponds to the period-3 orbit in Fig.~\ref{fig2}.
As shown in Fig.~\ref{fig2}, this subnetwork does not have any garden-of-Eden states.

\subsubsection{Subnetwork $B$}

For the subnetwork $B$ in Fig.~\ref{fig2}, we cannot diagonalize the corresponding Koopman submatrix $K_B$.
As shown in Fig~\ref{fig4}(b), we obtain five Koopman eigenvalues, $\lambda_{4} = 1$, $\lambda_{5,6,7,8} = 0$ (algebraic multiplicity $4$), and the four associated rank-$1$ eigenvectors of $K_B$,
\begin{align}
{\bm \phi_B}^{(4)}& =(\phi^{(4)}_1, \phi^{(4)}_4, \phi^{(4)}_6, \phi^{(4)}_7, \phi^{(4)}_8)^{\sf T} =(1, 1, 1, 1, 1)^{\sf T},
\cr
{\bm \phi_B}^{(5)}& =(\phi^{(5)}_1, \phi^{(5)}_4, \phi^{(5)}_6, \phi^{(5)}_7, \phi^{(5)}_8)^{\sf T}  = (0, 1, 0, 0, 0)^{\sf T},
\cr
{\bm \phi_B}^{(6)}& =(\phi^{(6)}_1, \phi^{(6)}_4, \phi^{(6)}_6, \phi^{(6)}_7, \phi^{(6)}_8)^{\sf T}  = (0, 0, 1, 0, 0)^{\sf T},
\cr
{\bm \phi_B}^{(7)}& =(\phi^{(7)}_1, \phi^{(7)}_4, \phi^{(7)}_6, \phi^{(7)}_7, \phi^{(7)}_8)^{\sf T}  = (0, 0, 0, 1, 0)^{\sf T},
\end{align}
and, additionally, one rank-$2$ generalized eigenvector
\begin{align}
 {\bm \phi_B}^{(8)} =(\phi^{(8)}_1, \phi^{(8)}_4, \phi^{(8)}_6, \phi^{(8)}_7, \phi^{(8)}_8)^{\sf T}  = (0, 0, 0, 0, 1)^{\sf T}.   
\end{align}
The rank-$1$ eigenvector ${\bm \phi}_B^{(4)}$ with the eigenvalue $\lambda_{4} = 1$ corresponds to the stationary state ${\bm x} = (0, 0, 0)$ ($q=1$) and all other states in its basin, i.e., ${\bm x} = (1, 1, 0)$, $(1, 0, 1)$, $(0, 1, 1)$, and $(1, 1, 1)$ ($q=4, 6, 7, 8$), and its components uniformly take a value $1$ at all $q \in \{1, 4, 6, 7, 8\}$.
The three rank-$1$ eigenvectors $\{{\bm \phi}_B^{(5)}, {\bm \phi}_B^{(6)}, {\bm \phi}_B^{(7)} \}$ associated with eigenvalues $\lambda_{5,6,7} = 0$ correspond to the three garden-of-Eden states, i.e., ${\bm x} = (1, 1, 0),\ (1, 0, 1),\ (0, 1, 1)$ ($q=4, 6, 7$), respectively, as can be seen from the locations of the component $1$.
The remaining rank-$2$ generalized eigenvector ${\bm \phi}^{(8)}_B$ corresponds to the state  ${\bm x} = (1, 1, 1)$ in the tree subnetwork, which connects the garden-of-Eden states and the stationary state.

\subsubsection{The whole network}

From the analysis for the subnetworks A and B, the original Koopman matrix $K$ has eight eigenvalues, $\lambda_{1,4} = 1$ (multiplicity $2$), $\lambda_{2} = e^{2 \pi i / 3}$, $\lambda_{3} = e^{4 \pi i / 3}$,  $\lambda_{5,6,7,8} = 0$ (multiplicity $4$), and associated seven rank-$1$ eigenvectors
\begin{align}
{\bm \phi}^{(1)} &=(0, \phi_{A,1}^{(1)}, \phi_{A,2}^{(1)}, 0, \phi_{A,3}^{(1)}, 0, 0, 0)^{\sf T} 
=
(0, 1, 1, 0, 1, 0, 0, 0)^{\sf T},\cr
{\bm \phi}^{(2)} &=(0, \phi_{A,1}^{(2)}, \phi_{A,2}^{(2)}, 0, \phi_{A,3}^{(2)}, 0, 0, 0)^{\sf T} 
=
(0, e^{\frac{2 i \pi }{3}}, e^{\frac{4 i \pi }{3}}, 0, 1, 0, 0, 0)^{\sf T},\cr
{\bm \phi}^{(3)} &=(0, \phi_{A,1}^{(3)}, \phi_{A,2}^{(3)}, 0, \phi_{A,3}^{(3)}, 0, 0, 0)^{\sf T} 
=
(0, e^{\frac{4 i \pi }{3}}, e^{\frac{2 i \pi }{3}}, 0, 1, 0, 0,  0)^{\sf T}.\cr
{\bm \phi}^{(4)} &= (\phi_{B,1}^{(4)}, 0, 0, \phi_{B,2}^{(4)}, 0, \phi_{B,3}^{(4)}, \phi_{B,4}^{(4)}, \phi_{B,5}^{(4)})^{\sf T}
= (1, 0, 0, 1, 0, 1, 1, 1)^{\sf T},\cr
{\bm \phi}^{(5)} &=(\phi_{B,1}^{(5)}, 0, 0, \phi_{B,2}^{(5)}, 0, \phi_{B,3}^{(5)}, \phi_{B,4}^{(5)}, \phi_{B,5}^{(5)})^{\sf T}
= (0, 0, 0, 1, 0, 0, 0, 0)^{\sf T},\cr
{\bm \phi}^{(6)} &=(\phi_{B,1}^{(6)}, 0, 0, \phi_{B,2}^{(6)}, 0, \phi_{B,3}^{(6)}, \phi_{B,4}^{(6)}, \phi_{B,5}^{(6)})^{\sf T}
= (0, 0, 0, 0, 0, 1, 0, 0)^{\sf T},\cr
{\bm \phi}^{(7)} &=(\phi_{B,1}^{(7)}, 0, 0, \phi_{B,2}^{(7)}, 0, \phi_{B,3}^{(7)}, \phi_{B,4}^{(7)}, \phi_{B,5}^{(7)})^{\sf T}
= (0, 0, 0, 0, 0, 0, 1, 0)^{\sf T},
\end{align}
 and one rank-$2$ eigenvector
\begin{align}
 {\bm \phi}^{(8)} =(\phi_{B,1}^{(8)}, 0, 0, \phi_{B,2}^{(8)}, 0, \phi_{B,3}^{(8)}, \phi_{B,4}^{(8)}, \phi_{B,5}^{(8)})^{\sf T}
= (0, 0, 0, 0, 0, 0, 0, 1)^{\sf T}.   
\end{align}
From the two Koopman eigenvectors with $\lambda_{1}=1$ and $\lambda_{4}=1$, we can construct the corresponding Koopman eigenfunctions, namely, two conserved quantities of the system.
As can be seen from Fig.~\ref{fig2}, each of these Koopman eigenvectors with the eigenvalue $1$ represents a connected component of the state-transition network.

We can summarize the above finding as follows. (i) Each eigenvector associated with the eigenvalue $1$ represents a connected component of the state-transition network; therefore, the multiplicity of the eigenvalue $1$ gives the number of connected components in the network. (ii) The eigenvalues reflect the period of the orbit embedded in each connected component.
(iii) The multiplicity of the zero eigenvalue gives the number of states that are not included in the periodic or stationary state.
(iv) If the system has no zero eigenvalue, the system has no garden-of-Eden state and is reversible.
In Sec. IV,  we explicitly construct these representative Koopman eigenvectors and rationalize the above observations.

\section{EXPLIXIT FORMS OF KOOPMAN EIGENFUNCTIONS}

\subsection{Eigenvalues and eigenvectors of the Koopman matrix}

In Sec.~III C, we have seen the Koopman eigenvalues and eigenvectors of rule 120 as an example.
As explained in Sec.~III B, for reversible ECA, the Koopman matrix is unitary and all Koopman eigenvalues are on the unit circle. 
For non-reversible ECA, Koopman eigenvalues $0$ also arise inside the unit circle as we saw for rule 120.

In general, the Koopman eigenvalues of finite ECA are either on the unit circle (including $1$) or $0$. This can be explained as follows.
The observable $\Phi$ associated with the eigenvalue $\lambda$ obeys $(K^n \Phi)({\bm x}) = \Phi({\bm F}^n({\bm x})) = \lambda^n \Phi({\bm x})$.
Therefore, if an eigenvalue $\lambda$ other than $0$ exists inside the unit circle, i.e., $0 < | \lambda | < 1$, the associated observable takes infinitely long time steps to converge to the final periodic or stationary values.
On the other hand, if an eigenvalue $\lambda$ with $| \lambda | > 1$ exists outside the unit circle, the associated observable $\Phi$ diverges.
However, both of these are not allowed in finite ECA.

In this section, we explicitly construct the Koopman eigenfunctions associated with these eigenvalues 
and discuss their relationship with reversibility and the number of conserved quantities of the system.

\subsection{Eigenvalue $0$}

For the indicator functions corresponding to the garden-of-Eden states, we have the following lemma.

\begin{lemma}
Assume that the system has a garden-of-Eden state ${\bm x}^{(u)} \in M$ with index $u \in \{1, ..., 2^N\}$. Then, the indicator function $b_{u}({\bm x})$
is a Koopman eigenfunction with the eigenvalue $0$; i.e., 
\begin{align}
(\hat{K} b_{u}) ({\bm x}) = b_{u}({\bm F}({\bm x})) = 0.
\end{align}
\end{lemma}
\begin{proof}
Because ${\bm x}^{(u)}$ is a garden-of-Eden state, no state ${\bm x}$ evolves into ${\bm x}^{(u)}$, i.e., ${\bm F}({\bm x}) \neq {\bm x}^{(u)}$ for all ${\bm x} \in M$; hence, $b_u({\bm F}({\bm x})) = 0$  for all ${\bm x} \in M$.
\end{proof}

The corresponding zero eigenvector of $K$ takes the form ${\bm \phi} = (0, ..., 0, 1, 0, ..., 0)^{\sf T}$,
where only the $u$th component takes $1$.
This can also be seen from the state-transition network $A$; because no link directed to $u$ exists, all components of the $u$th row of $A$ and $u$th column of $K = A^{\sf T}$ are zero; hence, the above ${\bm \phi}$ yields a zero vector when we multiply $K$, i.e., $K {\bm \phi} = 0$.

If the system has $k$ different garden-of-Eden states, the system has at least $k$ zero Koopman eigenvalues.
Therefore, if the system has no zero eigenvalue, the system has no garden-of-Eden state, i.e., the system is reversible.\\

The above lemma can be generalized to other indicator functions as follows. 

\begin{lemma}
Assume that the system has a state ${\bm x}^{(q)} \in M$ with index $q \in \{1, ..., 2^N\}$ that is not included in any of the periodic orbits (including stationary states). Then, the indicator function $b_q({\bm x})$ is a generalized Koopman eigenfunction of rank $m+1$ with the eigenvalue $0$, i.e.,
\begin{align}
( \hat{K}^{m+1} b_q) ({\bm x}) = b_q({\bm F}^{m+1}({\bm x})) = 0,
\end{align}
where $m \geq 0$ is the distance to the most distant garden-of-Eden state from ${\bm x}^{(q)}$ in the same tree subnetwork as ${\bm x}^{(q)}$. When ${\bm x}^{(q)}$ is a garden-of-Eden state, $m=0$.
\end{lemma}

\begin{proof}
The state ${\bm x}^{(q)}$ belongs to a tree subnetwork whose root state is included in the periodic orbit and whose leaves are garden-of-Eden states (Sec.II B).
Therefore, if the distance from ${\bm x}^{(q)}$ to the most distant garden-of-Eden state in the same tree subnetwork is $m$, we have ${\bm F}^{m+1}({\bm x}) \neq {\bm x}^{(q)}$ for all ${\bm x} \in M$; hence, $b_q({\bm F}^{m+1}({\bm x})) = 0$ for all ${\bm x} \in M$.
\end{proof}

The corresponding vector, ${\bm \phi} = (0, ..., 0, 1, 0, ..., 0)^{\sf T}$, where only the $q$th component takes $1$, is a generalized zero eigenvector of rank $m+1$ with the eigenvalue $0$ of $K$; i.e., $K^{m+1} {\bm \phi} = 0$.
The maximal rank of the generalized eigenfunction gives the maximum number of time steps required for the system to converge to the periodic orbit in a given connected component.

\subsection{Eigenvalue $1$}

For the indicator functions corresponding to isolated stationary states, we have the following lemma.

\begin{lemma}
Assume that a state ${\bm x}^{(s)} \in M$ with index $s \in \{1, ..., 2^N\}$ is an isolated stationary state of the system, i.e., ${\bm F}({\bm x}^{(s)}) = {\bm x}^{(s)} $ and ${\bm F}({\bm x}) \neq {\bm x}^{(s)}$ for all ${\bm x} \neq {\bm x}^{(s)}$.
Then, the indicator function $b_{s}({\bm x})$ is a Koopman eigenfunction associated with the eigenvalue $1$; i.e.,
\begin{align}
(\hat{K} b_{s}) ({\bm x}) =  b_{s}({\bm F}({\bm x})) = b_{s}({\bm x}).
\end{align}
\end{lemma}
\begin{proof}
If ${\bm x} \neq {\bm x}^{(s)}$, we have $b_s({\bm x}) = b_s({\bm F}({\bm x})) = 0$. If ${\bm x} = {\bm x}^{(s)}$, we have $b_s({\bm x}) = b_s({\bm F}({\bm x})) = 1$. Thus, the above eigenvalue equation holds for all ${\bm x} \in M$.
\end{proof}
The corresponding Koopman eigenvector is given in the form
${\bm \phi} = (0, ..., 0, 1, 0, ..., 0)^{\sf T}$, where only the $s$th component takes the value $1$ and all other components vanish.\\

The above lemma can be generalized to connected components (subnetworks) of the state-transition network.

\begin{lemma}
Let $S \subseteq M$ be a set of states in a connected component of the state-transition network and $Q_S \subseteq \{1, ..., 2^N\}$ the corresponding set of indices. Then, the observable
\begin{align}
c({\bm x}) = \sum_{q \in Q_S} b_q( {\bm x} )
\label{uniformeigenfunction}
\end{align}
is a Koopman eigenfunction associated with the eigenvalue $1$, i.e.,
\begin{align}
(\hat{K} c)({\bm x}) = c( {\bm F}( {\bm x} ) ) = c({\bm x}).
\end{align}
\end{lemma}

\begin{proof}
If ${\bm x} \in S$, we have ${\bm F}({\bm x}) \in S$ and $c({\bm x}) = c({\bm F}({\bm x})) = 1$. If ${\bm x} \notin S$, we have ${\bm F}({\bm x}) \notin S$ and $c({\bm x}) = c({\bm F}({\bm x})) = 0$. Thus, the above eigenvalue equation holds for all ${\bm x} \in M$.
\end{proof}

The vector components of the corresponding eigenvector ${\bm \phi} = (\phi_1, ..., \phi_{2^N})^{\sf T}$ are given by $\phi_q = 1$ for $q \in Q_S$ and $\phi_q = 0$ for $q \notin Q_S$.
\\

\subsection{Eigenvalue on the unit circle}

We can also construct the eigenfunctions corresponding to an isolated periodic orbit.
Assume that the system has an isolated period-$T$ orbit, represented as ${\bm x}^{(P_0)} \to {\bm x}^{(P_1)} \to ... \to {\bm x}^{(P_{T-1})} \to {\bm x}^{(P_T)} = {\bm x}^{(P_0)}$, where $P_0$, ..., $P_{T-1}$, $P_T = P_0$ are the indices of the states along the periodic orbit.
We denote the set of these states as $\chi$.
The system states not included in $\chi$ can never reach $\chi$; i.e., ${\bm F}({\bm x}) \notin \chi$ when ${\bm x} \notin \chi$.
We consider an observable of the form
\begin{align}
c({\bm x}) = \sum_{q=0}^{T-1} c_q b_{P_q}({\bm x}),
\end{align}
where $c_0, ..., c_{T-1}$ are non-zero coefficients and seek the condition that this $c({\bm x})$ is a Koopman eigenfunction with the eigenvalue $\lambda$; i.e., it satisfies the eigenvalue equation 
\begin{align}
\hat{K} c({\bm x}) = \lambda c({\bm x}).
\label{eigenvalueeq1}
\end{align}

First, if ${\bm x} \notin \chi$, we have $\hat{K} c({\bm x}) = c({\bm F}({\bm x})) = 0$ and $c({\bm x}) = 0$; thus the eigenvalue equation~(\ref{eigenvalueeq1}) is satisfied irrespective of $\lambda$.
Next, for ${\bm x} = {\bm x}^{(P_j)} \in \chi$ ($j=0, ..., T-1$), we have
\begin{align}
c({\bm x}) = \sum_{q=0}^{T-1} c_q b_{P_q}({\bm x}) = c_j
\end{align}
and
\begin{align}
\hat{K} c({\bm x}) = c({\bm F}({\bm x})) = \sum_{q=0}^{T-1} c_q b_{P_q}({\bm F}({\bm x})) = c_{j+1},
\end{align}
where we used ${\bm F}({\bm x}^{(P_j)}) = {\bm x}^{(P_{j+1})}$ and that the index $j+1$ is considered in modulo $T$.
Thus, to satisfy the eigenvalue equation~(\ref{eigenvalueeq1}), we need $c_{j+1} = \lambda c_j$, and hence,
$c_{j} = \lambda^j c_0$.
Because $c_{j+T} = c_j$, the eigenvalue $\lambda$ should satisfy $\lambda^T = 1$. Assuming $c_0 = 1$ without loss of generality, we obtain the following lemma. 

\begin{lemma}
Let $b_{P_0}({\bm x}), ..., b_{P_{T-1}}({\bm x})$ denote the indicator functions corresponding to an isolated periodic orbit ${\bm x}^{(P_0)} \to ... \to {\bm x}^{(P_{T-1})} (\to {\bm x}^{(P_0)})$. Then, the observable
\begin{align}
c({\bm x}) = \sum_{q=0}^{T-1} \lambda^q b_{P_q}({\bm x})
\end{align}
is a Koopman eigenfunction with the eigenvalue $\lambda = \exp(2\pi i k / T)$ ($k=0, 1, ..., T-1$).
\end{lemma}

Thus, for a period-$T$ orbit, we can construct $T$ independent eigenfunctions. In particular, when $k=0$, we obtain $\lambda = 1$ and $c_0 = c_1 = ... = c_{T-1} = 1$, which yields the Koopman eigenfunction with the eigenvalue $1$ discussed in Lemma 4 for the connected component of the state-transition network.\\

The above discussion can be generalized to a non-isolated periodic-$T$ orbit as follows. As before, we denote by  $\chi = \{ {\bm x}^{(P_0)}, {\bm x}^{(P_1)}, ... ,{\bm x}^{(P_{T-1})} \}$ the set of states included in the period-$T$ orbit. We consider a connected component $S = \{  {\bm x}^{(P_0)}, {\bm x}^{(P_1)}, ... ,{\bm x}^{(P_{T-1})}, {\bm x}^{(t_1)}, {\bm x}^{(t_2)}, ..., {\bm x}^{(t_s)} \}$ of the state-transition network in which $\chi$ is included, where $t_1, ..., t_{s}$ are the indices of the states in $S$ but not included in $\chi$.

\begin{lemma}
Let $S \subseteq M$ denote a connected component of the state-transition network including $\chi$ and $Q_S = \{ P_0, P_1, ..., P_{T-1}, t_1, t_2, ..., t_{s} \} \subseteq \{1, ..., 2^N \}$ the corresponding set of state indices, where $s$ is the number of states not included in $\chi$. Then, the observable
\begin{align}
c({\bm x}) = \sum_{q \in Q_S} \lambda^{-D_q} b_q({\bm x}),
\label{eq33}
\end{align}
where $D_q$ is the distance from ${\bm x}^{(q)}$ to ${\bm x}^{(P_0)} \in \chi$ along a directed path of the state-transition network and is a Koopman eigenfunction with the eigenvalue $\lambda = \exp(2\pi i k / T)$ ($k=0, 1, ..., T-1$).
\end{lemma}

\begin{proof}
If ${\bm x} \notin S$, we have ${\bm F}({\bm x}) \notin S$ and $c({\bm x}) = c({\bm F}({\bm x})) = 0$; therefore, the eigenvalue equation~(\ref{eigenvalueeq1}) holds. If ${\bm x} \in S$, we have
\begin{align}
c({\bm x}) = \sum_{q \in Q_S} \lambda^{-Dq} b_q({\bm x}) = \lambda^{-D},
\end{align}
where $D$ is the distance from ${\bm x}$ to ${\bm x}^{(P_0)}$, and
\begin{align}
\hat{K} c({\bm x}) =  \sum_{q \in Q_S} \lambda^{-Dq} b_q({\bm F}({\bm x})) 
= \lambda^{-D + 1},
\end{align}
which follows from the fact that ${\bm F}({\bm x})$ is one step closer to ${\bm x}^{(P_0)}$ than ${\bm x}$ along the directed path, and thus the distance from ${\bm F}({\bm x})$ to ${\bm x}^{(P_0)}$ is $D-1$.
Thus, the eigenvalue equation~(\ref{eigenvalueeq1}) is also satisfied.
Now, if we consider ${\bm x}$ in $\chi$, e.g., ${\bm x} = {\bm x}^{(P_0)}$, we have
$c({\bm x}^{(P_0)}) = 1 = c({\bm x}^{(P_T)}) = \lambda^{-T}$ because of the $T$-periodicity; hence, the eigenvalue $\lambda$ should be $\lambda = \exp(2\pi i k / T)$ ($k=0, 1, ..., T-1$).
\end{proof}

Similar to the  case of the isolated orbit, when the connected component has a period-$T$ orbit, we can construct $T$ independent eigenfunctions, now with non-zero components also on the states not included in the periodic orbit. Note that Lemma 5 is reproduced when the periodic orbit is isolated and $\chi = S$.

\subsection{Number of independent eigenfunctions}

Let us consider the total number of independent eigenfunctions constructed above. We focus on a single connected component of the state-transition network of size $T+U$, consisting of a period-$T$ orbit ($T$ states) and other states not included in the periodic orbit ($U$ states).
As mentioned in Sec.II B, a single connected component can include only a single periodic orbit, so any connected component of the network can be regarded as such.
For the eigenvalue $0$, we found that the indicator functions $b_q({\bm x})$ of the states $q$
not included in the periodic orbit are the (generalized) eigenfunctions; hence, there are $U$ independent eigenfunctions.
On the other hand, for the $T$ eigenvalues on the unit circle (including $1$), we obtained $T$ associated eigenfunctions which are mutually independent.
Thus, we obtained $T+U$ independent eigenfunctions for a single connected component of the network of size $T+U$; namely, we could obtain the complete set of eigenfunctions that span the space of observables defined in the connected component under consideration.
In particular, only one eigenfunction with the form of Eq.~(\ref{uniformeigenfunction}) associated with the eigenvalue $1$ exists for each connected component.
The same argument applies to all other connected components; hence, we can obtain the complete set of eigenfunctions to span the observables defined in the whole state-transition network.
From the above results, the following theorems follow:
\begin{theorem}
The algebraic multiplicity of the eigenvalue $0$ is equal to the number of system states that are not included in any periodic orbits. The system without eigenvalue $0$ is reversible.
\end{theorem}

\begin{theorem}
Each connected component corresponds to a single eigenfunction with the eigenvalue $1$.
Thus, the multiplicity of the eigenvalue $1$ gives the number of connected components in the whole state-transition network (and also the number of independent conserved quantities).
\end{theorem}

\begin{remark}
As stated in Sec.~III A, the eigenfunction $c ({\bm x})$ associated with the eigenvalue $1$ gives a conserved quantity of the system and can play important roles in characterizing the system dynamics.
If there are two or more independent eigenfunctions with eigenvalue $1$, linear sums of those eigenfunctions are also conserved quantities.
For example, for rule 184, the total number of black cells is conserved.
The corresponding Koopman eigenfunction $c_{black}({\bm x}) $ can be expressed as
\begin{equation}
c_{black}({\bm x}) = \sum_{s=1}^M N_s c_s ({\bm x}).
\end{equation}
Here, $M$ is the number of connected components in the network and hence the number of independent eigenfunctions with eigenvalue $1$, denoted by $\{ c_1,c_2,...,c_M \}$, and $N_s$ is the number of black cells in the states belonging to the connected component corresponding to the eigenfunction $c_s$.
\end{remark}

\begin{remark}
By using the Koopman eigenfunction $c({\bm x})$ associated with the eigenvalue $\lambda = \exp(2 \pi i / T)$ in Eq.~(\ref{eq33}), we can introduce the ``asymptotic phase'' of ECA, generalizing the notion used in continuous-time dynamical systems.
For an asymptotically stable limit cycle in continuous-time dynamical systems, the set of states that converge to the same state on the limit cycle is called an isochron and the same asymptotic phase~\cite{mauroy2013isochron,shirasaka2017phase}, typically taken in $[0, 2\pi)$, is assigned to it. It is known that the asymptotic phase can be represented by using the Koopman eigenfunction associated with the natural frequency of the limit cycle.
Similarly, we can define the isochron of ECA as the set of states that are equally distant from a reference state on the periodic orbit.
The asymptotic phase $\theta({\bm x}^{(q)})$ of the state ${\bm x}^{(q)} \in S$ whose distance to a reference state ${\bm x}^{(P_0)} \in \chi$ along the path is $D_q$ can be represented by using $c({\bm x})$ as 
\begin{align}
\theta({\bm x}^{(q)}) = i \log(c({\bm x}^{(q)})) = i \log \lambda^{-D_q} = i \log(\exp(-2\pi i D_q / T)) \in [0,2\pi),
\end{align}
where we take the principal value of the logarithm in the range $[0, 2\pi)$. We then have $\theta({\bm F}({\bm x}^{(q)})) - \theta({\bm x}^{(q)}) = 2 \pi / T$ for any ${\bm x}^{(q)} \in S$, namely, the asymptotic phase always increases with a constant frequency $2\pi/T$ as the system state evolves. 
\end{remark}

\begin{remark}
Though we focused on the Koopman operator $\hat{K}$ in this study, we can also consider the Perron-Frobenius operator~\cite{lasota2008probabilistic} $\hat{P}$, an adjoint operator of $\hat{K}$, and consider its spectral properties in a similar way to $\hat{K}$. See the Appendix for a brief explanation on the Perron-Frobenius operator.
\end{remark}

\section{RESULTS FOR ECA WITH 13 CELLS}

\subsection{Setup}

We now perform a thorough numerical analysis of all rules of ECA on a lattice of $13$ cells with periodic boundary conditions.
Since Wolfram's classification is for large systems, it is desirable to use a large number of cells for the numerical analysis.
However, as the total number $2^N$ of the system states grows exponentially fast with the number of cells $N$, numerical analysis of the Koopman matrix quickly becomes impossible.
In what follows, we use a lattice of $13$ cells with periodic boundary conditions, where $13$ was the largest prime number (chosen to avoid atypical size-dependent dynamics of ECA) that could be used to calculate the eigenvalues of $2^{13} \times 2^{13}$ matrices by a computer with 16GB memory.
For each independent ECA rule, we constructed the state-transition network and the Koopman matrix and numerically calculated all the Koopman eigenvalues.

\begin{figure}[h!]
    \centering
    \includegraphics[width=0.85\hsize]{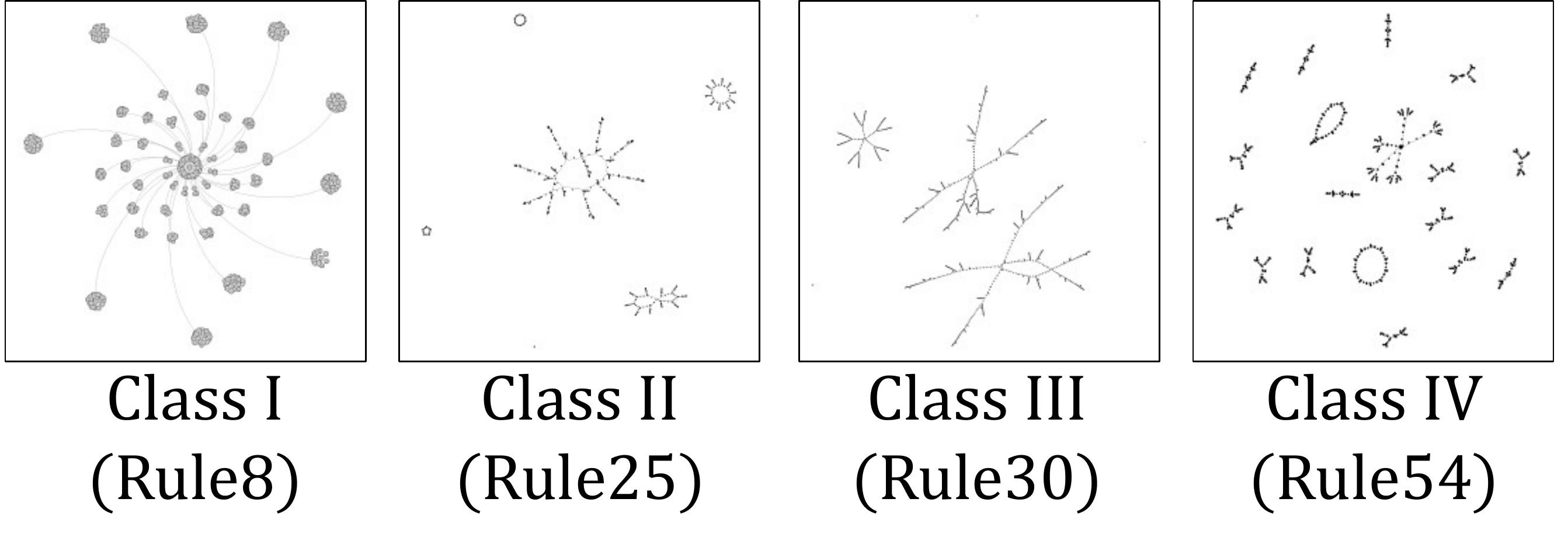}
    \caption{Typical transition networks of the four classes in ECA.}
    \label{fig5}
\end{figure}

\subsection{Properties expected for the Koopman eigenvalues}

Figure~\ref{fig5} shows typical state-transition networks for representative rules from Wolfram's four classes, i.e., rule 8 from class I, rule 25 from class II, rule 30 from class III, and rule 54 from class IV, for the case with $N=10$ cells. Here, because the state-transition networks with $N=13$ cells has too many nodes to be drawn, this figure shows the network for $N=10$ cells; the rest of the numerical calculations for the spectral properties were performed for $N=13$ cells.
We can observe that the network topology shows characteristic properties for each class.
From the asymptotic dynamical property of each class, the following properties for the Koopman eigenvalues are expected (zero eigenvalue does not arise in reversible ECA).
\begin{itemize}

\item Class I: As the system converges to a homogeneous stationary state, the Koopman eigenvalues take either $0$ (with large multiplicity) or $1$ (with small multiplicity).

\item Class II: As the system converges to a periodic orbit with a relatively short period, the Koopman eigenvalues take $T$ values, $\exp(2 \pi i k / T)$ where $T$ is the period and $k=0, 1, ..., T-1$, on the unit circle. Because orbits with various periods can coexist, the overall distribution of eigenvalues on the unit circle may not be even. Eigenvalues $0$ with large multiplicity also arise.

\item Classes III and IV: The system exhibits chaotic or complex dynamics. The orbits are still periodic, but their periods are typically much longer than the cases in class II. Thus, a large number of Koopman eigenvalues distribute on the unit circle, typically unevenly because orbits with different periods coexist.
Eigenvalues $0$ with large multiplicity also arise.

\end{itemize}

Figure~\ref{figall} shows the Koopman eigenvalues in the complex plane for all 256 (88 independent) rules of ECA with 13 cells and periodic boundary conditions. In the caption of each figure, Wolfram's class (I to IV) and the multiplicities of the eigenvalues $0$ and $1$ are shown. The eigenvalue $1$ corresponds to either a periodic or stationary state; hence, the number of stationary states is also shown.
The rules \{15, 85\}, 51, 204, \{154, 166, 180, 210\}, \{170, 240\}, \{45, 75, 89, 101\}, 105, and 150 do not possess zero eigenvalues and are reversible. These results coincide with the results obtained for systems of arbitrary sizes by de Brujin graph analysis for 13 cells~\cite{nobe2004reversible}.
We can observe that the distributions of the Koopman eigenvalues roughly follow our expectations above, but there are also some exceptions.
This is because the Koopman matrix has the whole information of the system dynamics, including non-typical ones, while Wolfram's classification is for typical dynamics of the system, and also because of the small system size used in the numerical analysis.
In Subsections V C-V E, we describe qualitative properties of the typical dynamics and eigenvalue distributions for each class.

\begin{figure}[h!]
    \centering
    \includegraphics[width=0.88\hsize]{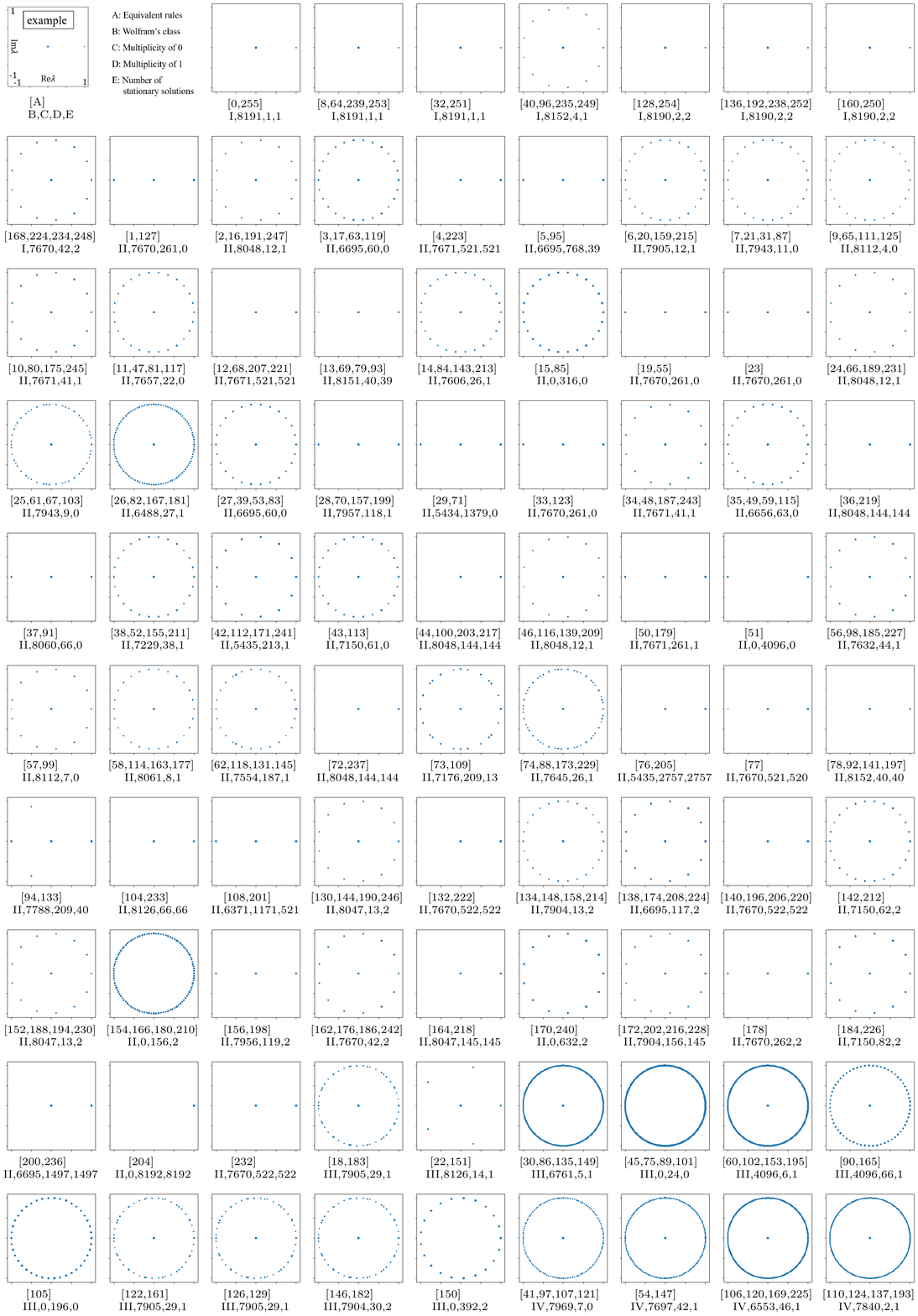}
    \caption{Koopman eigenvalues for all 256 rules of 13-cell ECA. The caption below each graph shows the set of  equivalent rules, class, and multiplicities of the $1$ and $0$ eigenvalues.}
    \label{figall}
\end{figure}

\begin{figure}[b]
    \centering
    \includegraphics[width=0.8\hsize]{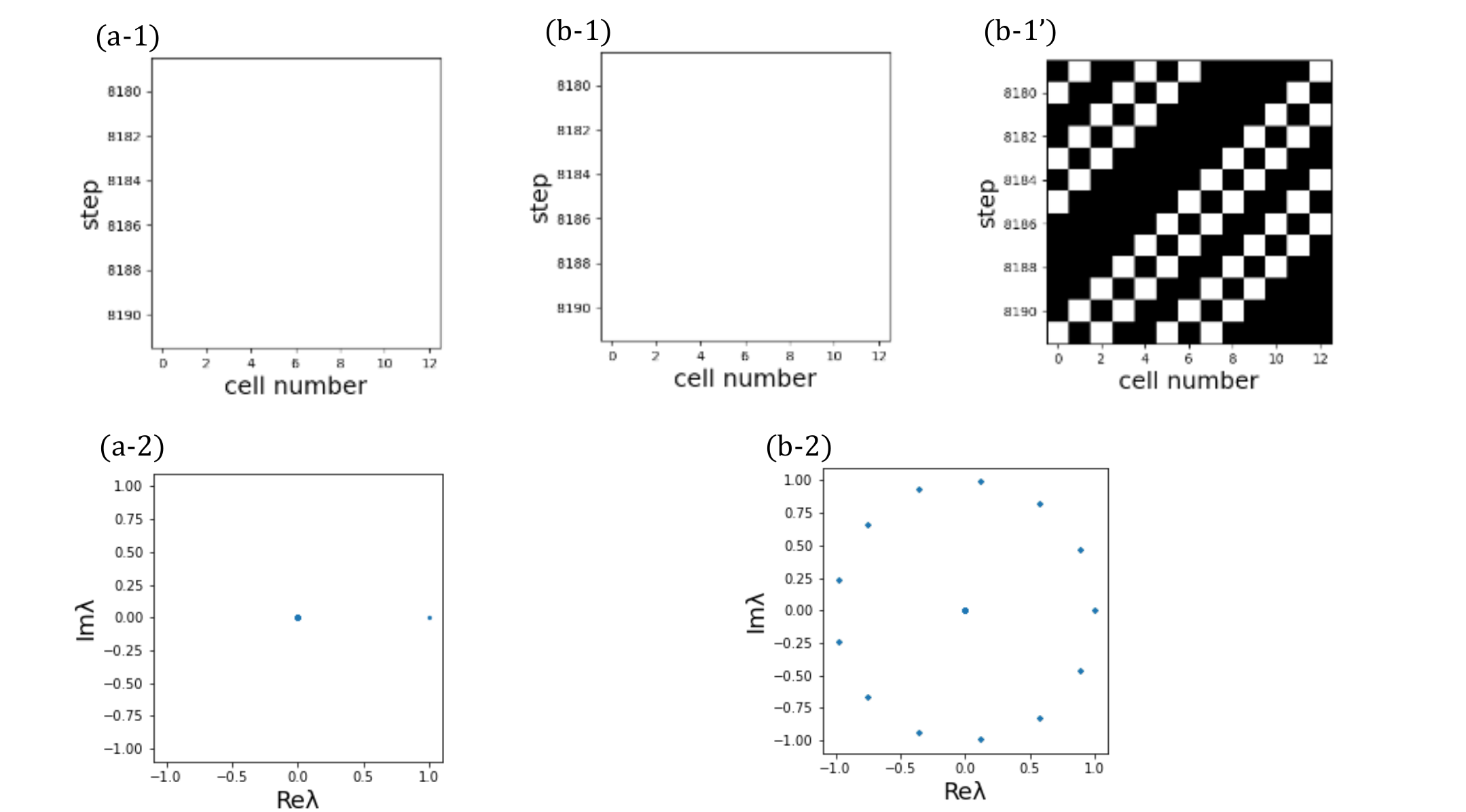}
    \caption{Dynamics and Koopman eigenvalues of class I rules. 
    (a) Rule 8. 
    (b) Rule 168. Results for typical homogeneous stationary states and for a non-typical traveling-wave state are shown.}
    \label{fig7}
\end{figure}

\subsection{Class I}

As shown in Fig.~\ref{figall}, for most of the rules, the Koopman eigenvalues are either $1$ (multiplicity $1$ or $2$) or $0$ (large multiplicity), as expected.
For the rules 40, 168, and their equivalents, the Koopman eigenvalues take $13$ values on the unit circle including $1$, in addition to $0$.
Indeed, these rules have non-typical traveling-wave states, corresponding to period-$13$ orbits, in addition to the typical homogeneous stationary state.
Figure~\ref{fig7} shows the dynamics and eigenvalues for rules 8 and 168 from class I, including a non-typical traveling-wave state of rule 168.

\subsection{Class II}

As shown in Fig.~\ref{figall}, for most of the rules, the Koopman eigenvalues evenly distribute on the unit circle (including $1$) or $0$. In contrast to the case of class I, the multiplicity of the eigenvalue $1$ can be relatively large.
The number of eigenvalues on the unit circle is typically up to $13$ and at most $104$.
Figure~\ref{fig8} shows typical dynamics and eigenvalues obtained for four rules from class II, which are briefly described below.

\begin{figure}[h!]
    \centering
    \includegraphics[width=\hsize]{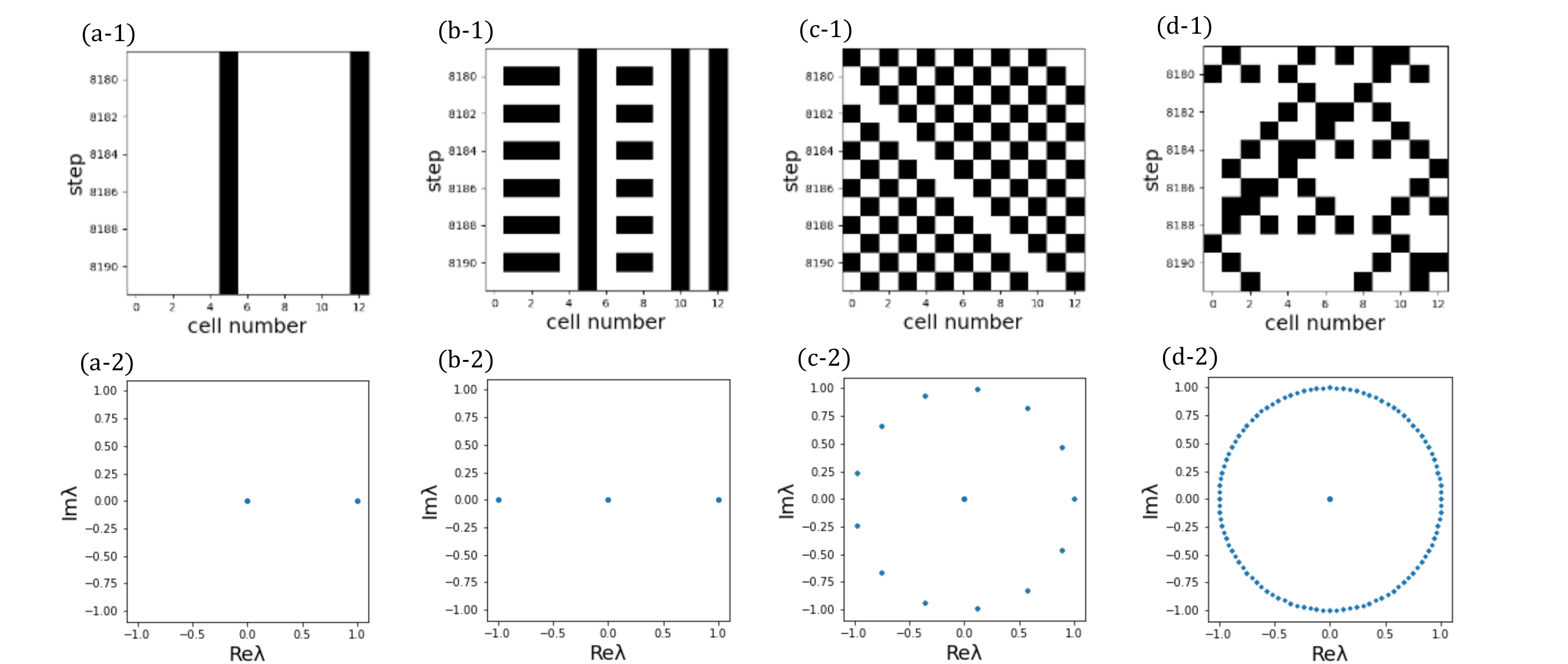}
    \caption{Dynamics and Koopman eigenvalues of class II rules. (a) Rule 4. (b) Rule 5. (c) Rule 184. (d) Rule 26.
}
    \label{fig8}
\end{figure}

\begin{list}{}{}

\item a) {\it Stationary states}. The system converges to a spatially inhomogeneous stationary state as shown in Fig.~\ref{fig8}(a) for rule 4, yielding Koopman eigenvalues $0$ and $1$. As the system can possess many different stationary inhomogeneous states (as well as periodic orbits), the multiplicity of the eigenvalue $1$ can be relatively large.

\item b) {\it Oscillatory states}. The system exhibits periodic oscillations of relatively short periods as shown in Fig.~\ref{fig8}(b) for rule 5, yielding Koopman eigenvalues $0$ and on the unit circle. As the system can possess many different periodic orbits (as well as stationary states), the multiplicity of the eigenvalue $1$ can be relatively large.

\item c) {\it Traveling-wave states}. The system's pattern translates to either of the directions without changing the shape as shown in Fig.~\ref{fig8}(c), which corresponds to a period-$13$ state. The Koopman eigenvalues take $13$ values on the unit circle in addition to $0$.

\item d) {\it Oscillatory-traveling states}. The system exhibits a mixture of traveling and oscillating patterns as shown in Fig.~\ref{fig8}(d) for rule 26, which travels to either of the directions while oscillating at the same time. The periods are between $26$ and 104. Relatively many eigenvalues arise on the unit circle (the maximum number of different Koopman eigenvalues on the unit circle is $104$ for rules $26$, $154$, and the equivalents).

\end{list}

\subsection{Classes III and IV}

As shown in Fig.~\ref{figall}, in most of the rules, at least $13$ eigenvalues appear on the unit circle.
In many rules, relatively many periodic orbits with different periods coexist, resulting in a large number of eigenvalues distributed on the unit circle. 
In some rules, small number (e.g., $5$ for rule 22) of eigenvalues appear on the unit circle, despite being in class III. This corresponds to relatively short-period oscillatory states, which can arise in small systems; in larger systems, the same rules can exhibit chaotic dynamics.
Figure~\ref{fig9} shows the dynamics and eigenvalues for rules 30 and 54 from class III and IV, respectively.
In these rules, the system exhibits complex patterns, which are periodic but  with very large periods. Thus, a large number of eigenvalues arise on the unit circle.

\begin{figure}[h!]
    \centering
    \includegraphics[width=0.55\hsize]{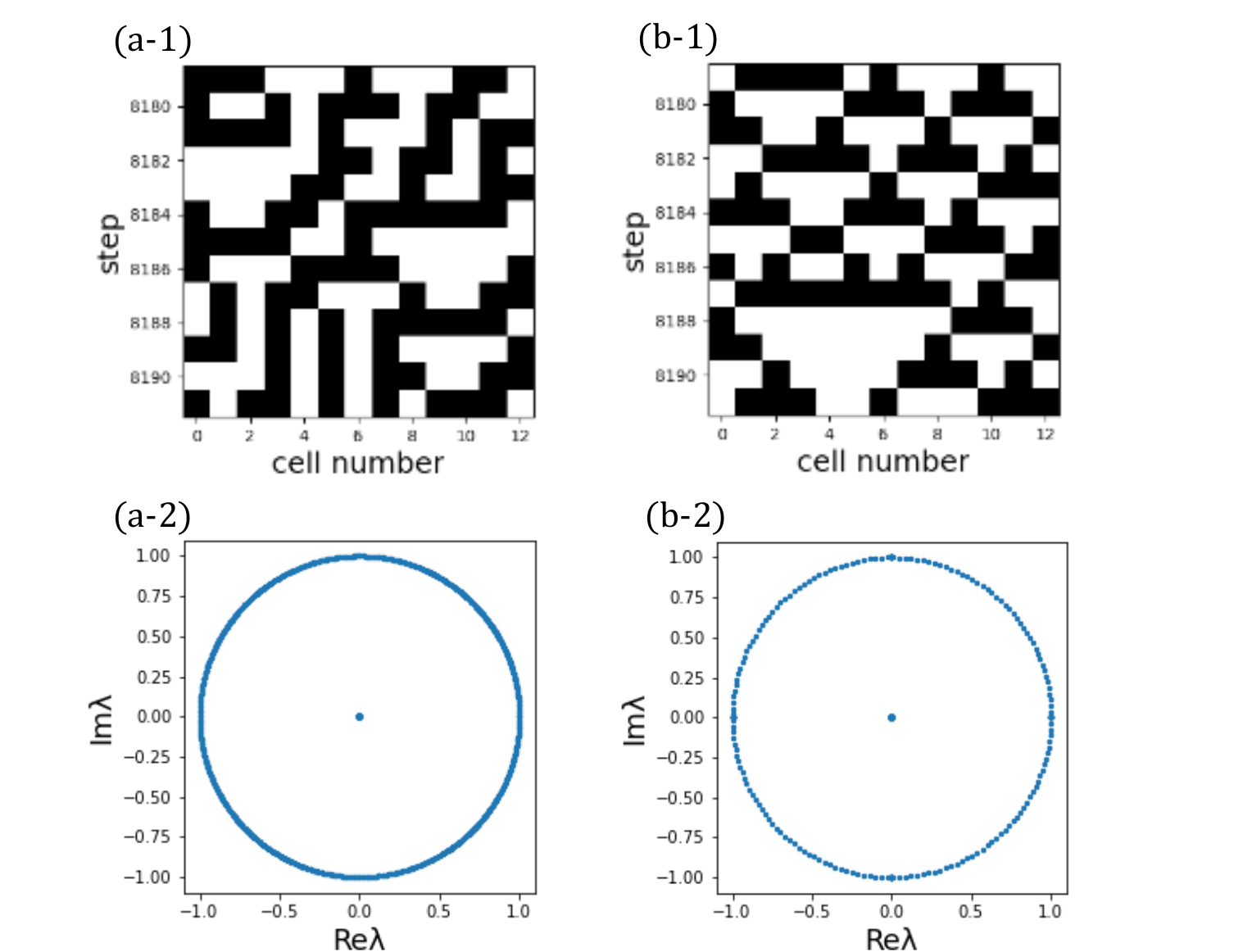}
    \caption{Dynamics and Koopman eigenvalues for class III and IV rules. (a) Rule 30, class III. (b) Rule 54, class IV.}
    \label{fig9}
\end{figure}

\section{DISCUSSION}

\subsection{Number of independent eigenvalues and lengths of periodic orbits}

In Sec. V, we numerically calculated the Koopman eigenvalues of $13$-cell ECA for all independent rules.
In our numerical analysis for several representative rules using different system sizes, the observed dependence of the number of different Koopman eigenvalues on the number of cells $N$ was typically as follows. 
(i) For the rules in class I and class II, only two or three different eigenvalues were obtained if asymptotic solutions were uniform or inhomogeneous stationary, and the number of the eigenvalues increased linearly with $N$ if traveling-wave solutions existed. The period of the system dynamics was at most $O(N)$.
(ii) For the rules with chaotic dynamics as those in class III and class IV, the number of independent eigenvalues could increase exponentially with $N$. The maximum period of the dynamics could also increase very rapidly.
From this observation, we conjecture that, in the limit $N \to \infty$, the set of different eigenvalues $\sigma$ has $card(\sigma)=\aleph_0$ for the rules in class I and class II, while $card(\sigma) = \aleph_1$ for the rules in class III and IV.
It is also conjectured that, in the limit $N \to \infty$, the maximum period of the dynamics diverges and the system can exhibit non-periodic chaotic dynamics characterized by densely distributed eigenvalues on the unit circle.
The scaling of the number of different eigenvalues with $N$ is thus expected to characterize the qualitative differences in the dynamical properties between the rules and may be used for the classification of ECA. 

\subsection{Dynamic Mode Decomposition}

In this study, we derived a matrix representation of the Koopman operator for ECA, which was equivalent to the full state-transition network of the ECA, and analyzed its spectral properties.
Owing to this representation, we could obtain the rigorous results in Sec. IV on the relationships between the dynamics of the ECA and the spectral properties of the Koopman operator that hold generally for any system size.
However, because the size of the Koopman matrix is $2^N \times 2^N$, numerical analysis of its spectral properties becomes quickly impractical as $N$ is increased.
Thus, the number of cells that we could use for the numerical analysis was $N=13$ in this study and it is difficult to increase this number largely.

The DMD and its extensions are standard, well-known methods for estimating the Koopman eigenvalues using  time-series data, which have been extensively used for the analysis of spatially extended systems such as fluid flows~\cite{schmid2010dynamic,rowley2009spectral}.
The DMD methods can also be used to estimate the Koopman eigenvalues of ECA.
In the case of the simplest DMD, we only need to construct an $N \times N$ matrix from the time-series data for estimating the Koopman eigenvalues, in contrast to the $2^N \times 2^N$ Koopman matrix; thus, the method is applicable to much larger systems.
Our preliminary analysis indicates that the DMD can reproduce the correct Koopman eigenvalues for relatively simple dynamics of ECA as long as $N$ is larger than the maximum period of the system.
When the maximum period exceeds $N$ but is still not too large, we can use generalized DMD methods, such as the Hankel DMD that uses delay-embedded time series, to estimate the correct Koopman eigenvalues.
Thus, the DMD methods provide alternative ways to estimate the Koopman eigenvalues for larger systems without resorting to the Koopman matrix.
However, our preliminary analysis also shows that some of the DMD methods, including the Hankel DMD, cannot always reproduce all of the Koopman eigenvalues correctly.
A systematic numerical investigation on the DMD analysis of ECA will be reported in our subsequent study.

In addition to DMD, we may be able to exploit the sparseness of the state-transition and Koopman matrices to develop more efficient methods that facilitate direct numerical analysis of the Koopman operator for large systems. Such approach would also be helpful for a more detailed analysis of the Koopman spectral properties of ECA.

\section{CONCLUSIONS}

We performed a Koopman spectral analysis of ECA. By introducing the one-hot representation, we derived a matrix representation of the Koopman operator. 
We showed that the Koopman eigenvalues are either zero or distributed on the unit circle (including $1$), and constructed the associated Koopman eigenfunctions. 
We then numerically calculated the Koopman eigenvalues of all $88$ independent rules for ECA with 13 cells and periodic boundary conditions.

The Koopman eigenvalues and their multiplicities reflect the topology of the state-transition network. The eigenvalues on the unit circle reflect the periods of the orbits embedded in the individual connected components.
The multiplicity of zero eigenvalues represent the number of states that are not included in any of the periodic orbits, and the reversibility of each rule can be judged from the non-existence of the eigenvalue $0$.
The Koopman eigenfunctions with the eigenvalue $1$ correspond to individual connected components of the state-transition network, and the multiplicity of the eigenvalue $1$ gives the number of connected components in the state-transition network and also the number of independent conserved quantities of the system.
We found that the distributions of the Koopman eigenvalues roughly correspond to Wolfram's classification, but there were also some exceptions.

A more detailed analysis on the relationships between the Koopman eigenvalues and the dynamical properties of ECA, together with systematic DMD analysis of larger systems, will be reported in our subsequent study.

\acknowledgments{We are grateful to Prof. Igor Mezi\'c for useful comments. We acknowledge JSPS KAKENHI JP17H03279, JP18H03287, JPJSBP120202201, JP20J13778, and JST CREST JP-MJCR1913 for financial support.}

\section*{DATA AVAILABILITY  STATEMENT}
The data that supports the findings of this study are available within the article.

\appendix

\section{PERRON-FROBENIUS OPERATOR}
Let us consider the adjoint of the Koopman operator, the Perron-Frobenius operator~\cite{lasota2008probabilistic}.
For ECA, the Perron-Frobenius operator $\hat{P}$ is defined by
\begin{align}
(\hat{P} g)({\bm x}) = \sum_{{\bm y} \in {\bm F}^{-1}({\bm x})} g({\bm y})
\end{align}
for a function $g : M \to {\mathbb C}$, where ${\bm y} \in {\bm F}^{-1}({\bm x})$ represents the set of states $\{{\bm y}\}$ satisfying ${\bm F}({\bm y}) = {\bm x}$.
The operation of $\hat{P}$ on the indicator function $b_q({\bm x})$ ($q=1, ..., 2^N$) can be expressed as follows. Let ${\bm x}^{(p)}$ the $p$th ($p=1, ..., 2^N$) state of the ECA. Then,
\begin{align}
(\hat{P} b_q)({\bm x}^{(p)}) = \sum_{{\bm y} \in {\bm F}^{-1}({\bm x}^{(p)})} b_q({\bm y}) = \sum_{r=1}^{2^N} b_q({\bm x}^{(r)}) A_{pr} = A_{pq}  = \sum_{r=1}^{2^N} A_{rq} b_r({\bm x}^{(p)}),
\end{align}
where $A$ is the state-transition matrix [$A_{pr} = 1$ if ${\bm F}({\bm x}^{(r)}) = {\bm x}^{(p)}$ and $A_{pr} = 0$ otherwise] and we used $b_q({\bm x}^{(r)}) = \delta_{q, r}$ ($\delta$ is the Kronecker delta). Since this holds for any $q, r \in \{1, ..., 2^N\}$, we have
\begin{align}
(\hat{P} b_q)({\bm x}) = \sum_{r=1}^{2^N} A_{rq} b_r({\bm x}) = b_s({\bm x}),
\label{PF_indicator}
\end{align}
where $s$ is the index of the state into which ${\bm x}^{(q)}$ evolves.
Expressing a general function $g$ as
\begin{align}
g({\bm x}) = \sum_{q=1}^{2^N} g_q b_q({\bm x}),
\end{align}
where $g_q \in \mathbb{C}$ is the expansion coefficient ($q=1, ..., 2^N$), the operation of $\hat{P}$ on $g$ is expressed as
\begin{align}
(\hat{P} g)({\bm x}) = \sum_{q=1}^{2^N} g_q (\hat{P} b_q)({\bm x}) 
=
\sum_{q=1}^{2^N} \left( \sum_{r=1}^{2^N} A_{rq} g_q \right) b_r({\bm x})
=
\sum_{q=1}^{2^N} \left( \sum_{r=1}^{2^N} A_{qr} g_r \right) b_q({\bm x}).
\label{PF0}
\end{align}
Thus, the matrix representation of $\hat{P}$ is simply given by the state-transition matrix $A$, which is the transpose $K^{\sf T}$ of the Koopman matrix $K$.

The Perron-Frobenius operator $\hat{P}$ and the Koopman operator $\hat{K}$ are adjoint to each other 
with respect to the inner product of two functions $f, g : M \to {\mathbb C}$ defined as
\begin{align}
\langle f, g \rangle = \sum_{q=1}^{2^N} f_q^* g_q,
\end{align}
namely,
\begin{align}
\langle \hat{K} f, g \rangle 
= \sum_{q=1}^{2^N} \left( \sum_{r=1}^{2^N} K_{qr} f^*_r \right) g_q 
= \sum_{r=1}^{2^N} f^*_r \left( \sum_{q=1}^{2^N} A^{\sf T}_{qr} g_q \right) 
= \sum_{q=1}^{2^N} f^*_q \left( \sum_{r=1}^{2^N} A_{qr} g_r \right) 
= \langle f, \hat{P} g \rangle.
\end{align}
Because $\hat{K}$ and $\hat{P}$ are finite-dimensional operators and adjoint to each other, they share the  same eigenvalues, but the associated eigenfunctions are different.

The following lemmas show that the eigenfunctions of the Perron-Frobenius operator $\hat{P}$ can also be expressed by using the indicator function $b_q({\bm x})$. The first lemma is for the eigenvalue $0$.

\begin{lemma}
Assume that the system has a state ${\bm x}^{(q)} \in M$ with the index $q \in \{1, ..., 2^N\}$ that is not included in any of the periodic orbits (including stationary states).
Let $m > 0$ be the distance from ${\bm x}^{(q)}$ to the periodic orbit to which it is attracted, and denote the point of arrival on the periodic orbit as ${\bm x}^{(s)} = {\bm F}^m({\bm x}^{(q)})$.
Let ${\bm x}^{(r)} $ be the state on the periodic orbit that reaches ${\bm x}^{(s)}$ in $m$ steps, i.e., ${\bm x}^{(s)} = {\bm F}^m({\bm x}^{(r)})$. 
Then, the function 
\begin{align}
c({\bm x})=b_q ({\bm x}) - b_r ({\bm x}) 
\end{align}
is a generalized Perron-Frobenius eigenfunction of rank $m$ with the eigenvalue $0$, i.e.,
\begin{align}
( \hat{P}^m c) ({\bm x}) = 0.
\end{align}
\end{lemma}

\begin{proof}

Since both ${\bm x}^{(q)}$ and ${\bm x}^{(r)}$ evolve into ${\bm x}^{(s)}$ in $m$ steps, from Eq.~(\ref{PF_indicator}),
\begin{align}
( \hat{P}^m c) ({\bm x}) = (\hat{P}^m b_q)({\bm x}) - (\hat{P}^m b_r)({\bm x}) = b_s({\bm x}) - b_s({\bm x}) = 0.
\end{align}
\end{proof}

The second lemma is for the eigenvalues on the unit circle.

\begin{lemma}
Assume that the system has a period-$T$ orbit $\chi = \{ {\bm x}^{(P_0)}, {\bm x}^{(P_1)}, ... ,{\bm x}^{(P_{T-1})} \}\subseteq M$, where ${\bm x}^{(P_T)} = {\bm x}^{(P_0)}$. 
Then, the function
\begin{align}
c({\bm x}) = \sum_{q =0}^{T-1} \lambda^{-q} b_{P_q}({\bm x})
\end{align}
is a Perron-Frobenius eigenfunction with the eigenvalue $\lambda = \exp(2\pi i k / T)$ ($k=0, 1, ..., T-1$), i.e.,
\begin{align}
( \hat{P} c) ({\bm x}) = \lambda c({\bm x}).
\end{align}
In particular, the eigenfunction with the eigenvalue 1 is simply given by 
\begin{align}
c({\bm x}) = \sum_{q=0}^{T-1} b_{P_q}({\bm x}).
\end{align}
\end{lemma}

\begin{proof}
Note that, for $q=0, ...,T-1$, $A_{P_q P_r} = \delta_{q,r+1}$ for $r=0, ..., T-2$ and $A_{P_q P_{r}} = \delta_{q,0}$ for $r=T-1$, because the states with indices $\{P_0, ..., P_{T-1}\}$ are on the period-$T$ orbit $\chi$.
Note also that $A_{q P_r} = 0$ when $q \notin \{P_0, ..., P_{T-1} \}$ for $r=0, ..., T-1$ because the states in $\chi$ cannot leave it.
By the definition of $\hat{P}$ and using the above properties of $A$, we have
\begin{align}
\hat{P} c({\bm x}) 
&= \sum_{q=0}^{T-1} \lambda^{-q} (\hat{P} b_{P_q})({\bm x}) 
= \sum_{q=0}^{T-1} \lambda^{-q} \sum_{r=0}^{T-1} A_{P_r P_q} b_{P_r}({\bm x})
= \sum_{q=0}^{T-1} \left( \sum_{r = 0}^{T-1} A_{P_q P_r} \lambda^{-r} \right) b_{P_q}({\bm x})
\cr
&=
\sum_{q = 0}^{T-1} \lambda^{-(q-1)} b_{P_q}({\bm x}) = \lambda \sum_{q = 0}^{T-1} \lambda^{-q} b_{P_q}({\bm x}) = \lambda c({\bm x}),
\end{align} 
where we used $\lambda^T = 1$. 
\end{proof}

Thus, similar to the Koopman eigenfunctions, we can construct all independent eigenfunctions of the Perron-Frobenius operator from the above lemmas.

\end{document}